\documentclass[conference]{IEEEtran}
\IEEEoverridecommandlockouts

\usepackage{balance}
\usepackage[linesnumbered,ruled,vlined]{algorithm2e}
\usepackage{multirow}
\usepackage{array}
\usepackage{varioref}
\usepackage{amsthm}
\usepackage{pifont}
\usepackage{tikz}
\usepackage{booktabs}
\usepackage{balance}
\usepackage[utf8]{inputenc}
\DeclareUnicodeCharacter{2212}{-}

\usetikzlibrary{matrix,chains,positioning,decorations.pathreplacing,arrows,backgrounds}
\usepackage{float}
\usepgflibrary{shapes.multipart}
\usepackage{tikz}
\usepackage{pgfplots}
\usepackage{subfigure}
\usepackage{pgf-umlsd}
\usepackage{footnote}
\usepackage{url}
\usepackage[most]{tcolorbox}
\usepackage[export]{adjustbox}
\usepackage{hyperref}
\newtheorem{proposition}{Proposition}

\usepackage{color, colortbl}
\definecolor{Gray}{gray}{0.9}
\definecolor{airforceblue}{rgb}{0.36, 0.54, 0.66}
\definecolor{aliceblue}{rgb}{0.94, 0.97, 1.0}
\definecolor{alizarin}{rgb}{0.82, 0.1, 0.26}
\definecolor{amber}{rgb}{1.0, 0.75, 0.0}
\definecolor{amber(sae/ece)}{rgb}{1.0, 0.49, 0.0}
\definecolor{awesome}{rgb}{1.0, 0.13, 0.32}
\definecolor{babypink}{rgb}{0.96, 0.76, 0.76}
\definecolor{bronze}{rgb}{0.8, 0.5, 0.2}
\definecolor{battleshipgrey}{rgb}{0.52, 0.52, 0.51}
\definecolor{bole}{rgb}{0.47, 0.27, 0.23}
\definecolor{bulgarianrose}{rgb}{0.28, 0.02, 0.03}
\definecolor{brinkpink}{rgb}{0.98, 0.38, 0.5}
\definecolor{cadet}{rgb}{0.33, 0.41, 0.47}
\definecolor{ceil}{rgb}{0.57, 0.63, 0.81}
\definecolor{cerulean}{rgb}{0.0, 0.48, 0.65}
\definecolor{charcoal}{rgb}{0.21, 0.27, 0.31}
\definecolor{coolblack}{rgb}{0.0, 0.18, 0.39}
\definecolor{coolgrey}{rgb}{0.55, 0.57, 0.67}
\definecolor{darkcandyapplered}{rgb}{0.64, 0.0, 0.0}
\definecolor{darkbrown}{rgb}{0.4, 0.26, 0.13}
\definecolor{darkcerulean}{rgb}{0.03, 0.27, 0.49}
\definecolor{darkgray}{rgb}{0.66, 0.66, 0.66}
\definecolor{darkgoldenrod}{rgb}{0.72, 0.53, 0.04}
\definecolor{darkjunglegreen}{rgb}{0.1, 0.14, 0.13}
\definecolor{darktaupe}{rgb}{0.28, 0.24, 0.2}
\definecolor{davy\'sgrey}{rgb}{0.33, 0.33, 0.33}
\definecolor{frenchblue}{rgb}{0.0, 0.45, 0.73}
\definecolor{almond}{rgb}{0.94, 0.87, 0.8}
\definecolor{beaublue}{rgb}{0.74, 0.83, 0.9}
\definecolor{beige}{rgb}{0.96, 0.96, 0.86}
\definecolor{bisque}{rgb}{1.0, 0.89, 0.77}
\definecolor{black}{rgb}{0.0, 0.0, 0.0}
\definecolor{fluorescentorange}{rgb}{1.0, 0.75, 0.0}
\definecolor{ghostwhite}{rgb}{0.97, 0.97, 1.0}
\definecolor{antiquewhite}{rgb}{0.98, 0.92, 0.84}
\definecolor{ao(english)}{rgb}{0.0, 0.5, 0.0}

\SetKwInput{KwInput}{Input}                
\SetKwInput{KwOutput}{Output} 
\newtheorem{myAttack}{Attack}

\begin{document}

\title{Learning in the Dark: Privacy-Preserving Machine Learning using Function Approximation
\thanks{This work was funded by the Technology Innovation Institute (TII) for the project ARROWSMITH and from Horizon Europe for HARPOCRATES (101069535).}
}

\author{\IEEEauthorblockN{1\textsuperscript{st} Tanveer Khan}
\IEEEauthorblockA{\textit{ Department of Computing Sciences} \\
\textit{Tampere University}\\
Tampere, Finland \\
tanveer.khan@tuni.fi}
\and
\IEEEauthorblockN{2\textsuperscript{nd} Antonis Michalas}
\IEEEauthorblockA{\textit{ Department of Computing Sciences} \\
\textit{Tampere University, Finland} and\\
\textit{RISE Research Institutes of Sweden}\\
antonios.michalas@tuni.fi}
}

\maketitle

\begin{abstract}
Over the past few years, a tremendous growth of machine learning was brought about by a significant increase in adoption and implementation of cloud-based services. As a result, various solutions have been proposed in which the machine learning models run on a remote cloud provider and not locally on a user's machine. However, when such a model is deployed on an untrusted cloud provider, it is of vital importance that the users' privacy is preserved. To this end, we propose Learning in the Dark -- a hybrid machine learning model in which the training phase occurs in plaintext data, but the classification of the users' inputs is performed directly on homomorphically encrypted ciphertexts. To make our construction compatible with homomorphic encryption, we approximate the ReLU and Sigmoid activation functions using low-degree Chebyshev polynomials. This allowed us to build Learning in the Dark -- a privacy-preserving machine learning model that can classify encrypted images with high accuracy. Learning in the Dark preserves users' privacy since it is capable of performing high accuracy predictions by performing computations directly on encrypted data. In addition to that, the output of Learning in the Dark is generated in a blind and therefore privacy-preserving way by utilizing the properties of homomorphic encryption. 
\end{abstract}


\section{Introduction}
\label{sec:intro}

Machine Learning (ML), specifically Deep Learning (DL), has garnered significant attention from researchers due to its solid performance in many tasks, such as speech recognition, spam detection, image classification, traffic analysis, face recognition, financial detection, and genomics prediction~\cite{rane2009secure, islam2011application, kim2015private, gilad2016cryptonets, dowlin2017manual, shortell2019secure}. 
To meet the growing demand for ML services, Cloud Service Providers (CSPs) such as Google Prediction API~\cite{Google}, Microsoft Azure ML~\cite{Microsoft}, and Ersatz Lab~\cite{Ersatz} also offer Machine Learning as a Service (MLaaS), enabling users to train and test the ML models using the CSP infrastructure. Typically, these models involve a training phase where the model learns from a dataset and a testing phase where the model predicts outputs based on unseen inputs. Once the model is trained and deployed on the CSP, the users can use it for online prediction services. 
However, the adoption of MLaaS raises concerns about the privacy of data being outsourced, in sensitive domains such as finance and healthcare~\cite{Michalas:14:Healthcom}. There is a risk of data misuse or theft when sending data to prediction models hosted by CSPs. To address these privacy concerns, researchers proposed various methods to protect user data in MLaaS settings~\cite{Patricia, garge2018neural, gilad2016cryptonets, khan2023split, khan2023love}. 

This work aims to demonstrate the application of Neural Network (NN) on Encrypted Data (ED) using Homomorphic Encryption (HE). HE allows performing arithmetic operations (addition and multiplication) over ED without decryption, enabling the homomorphic evaluation of functions relying on these operations. More specifically, our focus is to evaluate the Convolution Neural Network (CNN) on ED, where most operations, except for Non-linear Activation Functions (NLAF), can be homomorphically evaluated. 

Enabling the homomorphic evaluation of CNNs on ED has been an active area of research, with significant efforts dedicated to designing efficient support for NLAFs~\cite{khan2021blind}. Various approaches have been proposed, including the utilization of power functions~\cite{gilad2016cryptonets}, look-up table~\cite{crawford2018doing}, and polynomial approximations~\cite{hesamifard2016cryptodl, chabanne2017privacy, chou2018faster}. In this work, we employ low-degree Chebyshev polynomials to approximate NLAF.

\subsection{Background on Polynomial Approximations}
\label{subsec:polynomial approximation}
Approximating continuous functions is a problem that has drawn mathematicians' attention for a very long time. While there are several ways to approximate a continuous function, in this work we are only interested in polynomial approximations. More specifically, we are using Chebysev polynomials to approximate the Sigmoid and the ReLU functions. However, there are various works that use different approaches such as the $x^2$ function~\cite{gilad2016cryptonets}, the \textit{Piecewise} approximation~\cite{chabanne2017privacy}, lookup tables~\cite{crawford2018doing} etc. Unfortunately, all these methods face certain limitations. For example, the $x^2$ method can cause instability during the training phase and the creation of a piecewise linear approximation can sometimes be a complex optimization problem. With this in mind, we chose to work with Chebyshev polynomials. The general form of these polynomials is:  

\begin{equation}
\label{equ:chebyshev}
T_{n+1}(x)=2xT_{n}(x)-T_{n-1}(x)
\end{equation}
where $T_{n}(x)$ represents a polynomial of degree $n$. Chebyshev polynomials allow us to efficiently compute any continuous function in a given interval, using only low-degree polynomials. This feature significantly boosts efficiency and lower the overall computational complexity.

\subsection{Our Contribution}
\label{subsec:our contribution}
The main contributions of this paper are manifold. 

\begin{itemize}
	\item We show how to approximate NLAFs like ReLU and Sigmoid using Chebyshev polynomials. By substituting these NLAFs with the Chebyshev polynomials, we conduct a comprehensive analysis to compare the differences in terms of efficiency and accuracy.  
	\item We design a PPML model in which the CNN is trained on plaintext data while the classification process operates on homomorphically ED.
	\item To illustrate the effectiveness of our model, we conducted extensive experiments and provided a comparative analysis with other state-of-the-art works in the field of PPML.
	\item We designed a protocol that demonstrate the practical application of our PPML model in a realistic scenario while ensuring its security under malicious threat model. 	
\end{itemize}

\subsection{Organization}
\label{subsec:organization}
The rest of the paper is organized as follows: In Section~\ref{sec:relatedwork}, we present important published works in the area of PPML. In Section~\ref{sec:preliminaries}, we provide the necessary background information needed for our construction. Then, in Section~\ref{sec:chebyshev approximation}, we show how to approximate the ReLU and Sigmoid AFs using low-degree Chebysev polynomials. The methodology of our work is illustrated in Section~\ref{sec:methodology}, followed by extensive experimental results in Section~\ref{sec:performance analysis}. In Section~\ref{sec:PPMLProtocol}, we design a protocol, that demonstrates the applicability of our work and finally, in Section~\ref{sec:Impact} we conclude the paper.

\section{Related Work}
\label{sec:relatedwork}

The first step in preserving the privacy of the ML model is achieved through Multiparty Computation (MPC). This approach allows parties jointly compute a function while keeping the original inputs private. Several methods based on MPC have been proposed for preserving the privacy of ML models, such as K-means clustering, linear regression, SVM classifier, Decision tree, etc.~\cite{bunn2007secure, sanil2004privacy, lindell2000privacy, vaidya2008privacy, slavkovic2007secure}. 



One approach called SecureML, designed by Mohassel \textit{et al.,}~\cite{mohassel2017secureml}, uses a two-server model in which data is distributed among two non-colluding servers. It is an efficient protocol for preserving the privacy of various ML models using MPC. These servers train various models on the joint data using secure MPC with support for approximating Activation Functions (AF) during training. Since SecureML requires changes in the training phase, the model does not apply to the problem of making the existing NN model oblivious. Another approach MiniONN~\cite{liu2017oblivious}, converts any NN into an oblivious NN using MPC providing privacy-preserving predictions. While MiniONN uses cryptographic primitives, such as garbled circuits and secret sharing, it still reveals information about the network (e.g. size of the filter)~\cite{juvekar2018gazelle}. MOBIUS is another secure prediction protocol for binarized NN~\cite{kitai2019mobius}, allowing fast and scalable PPML model by delegating a protected model to a resource provider. The resource provider offers prediction to client without knowing anything about client's input.  

Due to the high communication cost associated with MPC techniques mentioned above, alternative methods using HE have been explored. Wu \textit{et al.}~\cite{wu2013privacy} proposed a privacy-preserving logistic regression model. As the logistic function is \textit{not} linear, the authors use polynomial fitting to achieve a good approximation. However, it lowers the accuracy of model. Graepel \textit{et al.}~\cite{graepel2012ml} used Somewhat Homomorphic Encryption (SHE)~\cite{fan2012somewhat} to train two simple classifiers 
on ED, employing low-degree polynomials for efficient computations. 

Ehsan \textit{et al.}~\cite{hesamifard2017cryptodl} proposed a technique based on Leveled HE (LHE)~\cite{brakerski2014leveled} to preserve the privacy of CNN while at the same time keep the accuracy as close as possible to the original model. They approximated the Sigmoid, ReLU and Tanh functions and achieved an accuracy of~99.52\% on MINST dataset~\cite{lecun-mnisthandwrittendigit-2010}. This is good, as the accuracy of the original model was measured at~99.56\%. Unfortunately, their approach is computationally expensive, as the training and testing phases are both performed on ED.

In~\cite{bourse2018fast}, the authors present Fast Homomorphic Evaluation of Deep Discretized NN (FHE-DiNN). Their design utilizes HE to evaluate an NN. The user encrypts the data using HE and transfers it to the cloud. The cloud blindly classifies the ED using HE and sends the ED back to the user. Upon reception, the user uses her secret key to decrypt it. In this scheme, the encryption parameters are dependent on the model structure. So, if the server updates its model, the client is forced to re-encrypt all of its data. While communication-wise HE schemes are very efficient, the computation cost at the server-side is very large.




A notable related work is CryptoNets~\cite{gilad2016cryptonets} which applies an NN model to ED. While CryptoNets achieves remarkable accuracy, the construction is based on the use of square AF. Hence, approximating a non-linear function causes instability during the training phase when the interval\footnote{By interval we mean the domain of definition of the AF.} is large.  In our work, we address this issue by using Chebyshev approximation, which accurately approximates AFs even in larger intervals. Additionally, we adapt an approach where the client's input is encrypted but the model remains in plaintext, aiming for better efficiency in the classification process.

\section{Preliminaries}
\label{sec:preliminaries}

\paragraph*{\textbf{Notation}} If $x$ and $y$ are two strings, by $x||y$ we denote the concatenation of $x$ and $y$. A \textit{probabilistic polynomial time} (PPT) adversary $\mathcal{ADV}$ is a randomized algorithm for which there exists a polynomial $p(\cdot)$ such that for all input $x$, the running time of $\mathcal{ADV}(x)$ is bounded by $p(|x|)$. A neuron is a mathematical function that  takes one or more inputs, multiplies them by some values called ``weights'' and adds them together. This value is then passed to NLAF, to become neuron's output.

\subsection{Convolutional Neural Network (CNN)}
\label{subsec:convolutional neural network}
A typical NN is a combination of neurons arranged in layers. Each neuron receives input from other neurons with an associated weight $w$ and a bias $b$, as shown in~\autoref{fig:neural network}. It then uses equation~\ref{equ:simple neural network} to compute some function $f$ on the weighted sum of its input. The output of this neuron is given as input to other neurons. 

\begin{equation}
	y = f \left( \sum_{i=1}^{3} {x_{i}w_i}+b \right) 
	\label{equ:simple neural network}
\end{equation}

In equation~\ref{equ:simple neural network}, $x_i$ is the input, $w_i$ is the weight, $b$ is the bias term and $f$ is the AF. 
  
\begin{figure}
	\resizebox{9cm}{3.5cm}{%
		\begin{tikzpicture}[
		init/.style={
			draw,
			circle,
			inner sep=2pt,
			font=\Huge,
			join = by -latex
		},
		squa/.style={
			draw,
			inner sep=2pt,
			font=\Large,
			join = by -latex
		},
		start chain=2,node distance=13mm
		]
		\node[on chain=2] 
		(x2) {$x_2$};
		\node[on chain=2,join=by o-latex] 
		{$w_2$};
		\node[on chain=2,init] (sigma) 
		{$\displaystyle\Sigma$};
		\node[on chain=2,squa,label=above:{\parbox{2cm}{\centering Activation \\ function}}]   
		{$f$};
		\node[on chain=2,label=above:Output,join=by -latex] 
		{$y$};
		\begin{scope}[start chain=1]
		\node[on chain=1] at (0,1.5cm) 
		(x1) {$x_1$};
		\node[on chain=1,join=by o-latex] 
		(w1) {$w_1$};
		\end{scope}
		\begin{scope}[start chain=3]
		\node[on chain=3] at (0,-1.5cm) 
		(x3) {$x_3$};
		\node[on chain=3,label=below:Weights,join=by o-latex] 
		(w3) {$w_3$};
		\end{scope}
		\node[label=above:\parbox{2cm}{\centering Bias \\ $b$}] at (sigma|-w1) (b) {};
		
		\draw[-latex] (w1) -- (sigma);
		\draw[-latex] (w3) -- (sigma);
		\draw[o-latex] (b) -- (sigma);
		
		\draw[decorate,decoration={brace,mirror}] (x1.north west) -- node[left=10pt] {Inputs} (x3.south west);
		\end{tikzpicture}}
	\caption{Structure of a Neuron in a Neural Network} \label{fig:neural network}
\end{figure}
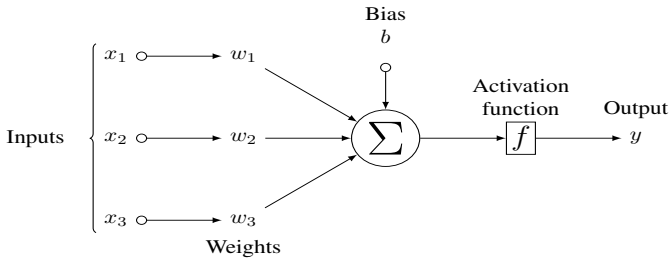

In our work, we focus on CNN, a deep NN algorithm primarily used for \textit{image classification}. In CNN, each input passes through a series of layers during the training and testing phases\footnote{\url{https://shorturl.at/nzHK1}}. These layers consist of convolutional layers (Conv), AFs, pooling layers, Fully Connected (FC) layers and a softmax layer~\cite{albawi2017understanding}:

\begin{figure}[h]
	\centering
	\includegraphics[width=\linewidth]{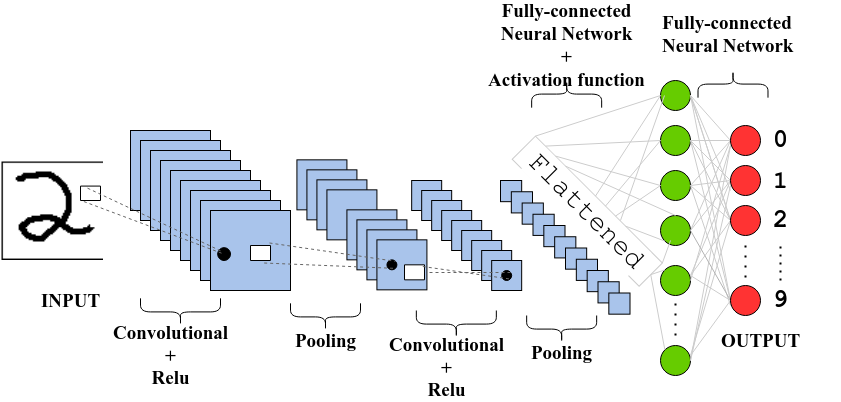}
	\caption{Convolutional Neural Network}
	\label{fig:Convolutional NN}
\end{figure}

\begin{itemize}
	\item \underline{Convolution Layer}: Conv is the first layer in CNN and acts as a feature detector (used for feature mapping). To generate a feature map, convolution is performed by moving the filter over the input with a certain stride. On a single input, multiple convolutions can be performed using numerous filters to extract more than one feature from the input. Also, padding is performed to make size of convolved features same as that of the input. 


	\item \underline{Activation Function}: In NN, all operations are linear except the AF. These functions are used to introduce non-linearity in the network. The most commonly used AFs are Sigmoid, ReLU and Tanh, as shown in the~\autoref{tab:diff activation ftn}. 
	
	\item \underline{Pooling}: This layer is responsible for extracting the dominant features (maximum or average pixel values) to reduce the size of the input image. The popular pooling operations are \textit{max-pooling} and \textit{average-pooling}. In max-pooling, the maximum value, and in average pooling, the average values are extracted from the part of the image covered by the filter.
	
	\item \underline{Flattening}: It convert data into a 1-dimensional array that is given as input to the next layer. There, the image matrix is converted into a vector and feed to a FC NN.
	
	\item \underline{Fully Connected}: The FC layers are activated at the last phase of the process after Conv, pooling and AFs. This layer connects every neuron in one layer to every neuron in the next layer and performs a weighted sum of the inputs and add bias.
	
	\item \underline{Softmax}: In a classification problem, softmax is the final output layer with discrete class labels. It assigns a probability to each class that adds up to 1. The node having the highest probability is determined to be the most likely class for the given input.

\end{itemize}

Our CNN (~\autoref{fig:Convolutional NN}) consists of two Conv with a ReLU AF, two pooling operations, two FC layers and a softmax layer.

\begin{table}[h]
	\centering
		\caption{Different Activation Functions.}
	\label{tab:diff activation ftn}
 \scriptsize
	\begin{tabular}{|p{10mm}|p{15mm}|p{15mm}|p{20mm}|}
		\hline
		\multicolumn{1}{|l|}{Name} & \multicolumn{1}{l|}{Function} & \multicolumn{1}{l|}{Derivative} & \multicolumn{1}{l|}{Figure} \\ 
		\hline
		ReLU & $f(x) =\begin{cases}
		0 & ~\text{if}~ x<0 \\ 
		x & ~\text{if}~x \geq 0.
		\end{cases}$ & $f'(x)=\begin{cases}
		0 & ~\text{if}~ x<0 \\ 
		x & ~\text{if}~1 \geq 0.
		\end{cases} $ & 
		\begin{tikzpicture}[baseline={(0,0.5)}]
		\draw (-1,0) -- (1,0);
		\draw (0,0) -- (0,1);
		\draw plot[domain=-1:1,variable=\x] ({\x},{ifthenelse(\x<0,0,\x)});
		\end{tikzpicture}\\Sigmoid & $f(x)=\frac{1}{1+e^{-x}}$ & $f'(x)=f(x)(1-f(x))^2$  &  
		\begin{tikzpicture}[baseline={(0,0.2)}]
		\draw (-1,0) -- (1,0);
		\draw (0,0) -- (0,1);
		\draw plot[domain=-1:1,variable=\x] ({\x},{1/(1+exp(-4*\x))});
		\end{tikzpicture}\\
		\\
		Tanh & $f(x)=\frac{e^x-e^{-x}}{e^z+e^{-z}} $ & $f'(x)=1-f(x)^2$   
		&  \begin{tikzpicture}[baseline={(0,0)}]
		\draw (-1,0) -- (1,0);
		\draw (0,-1) -- (0,1);
		\draw plot[domain=-1:1,variable=\x] ({\x},{tanh(4*\x)});
		\end{tikzpicture} \\
		\hline                          
	\end{tabular}
\end{table}

\subsection{Homomorphic Encryption (HE)}
\label{subsec:homomorphic encryption}
HE is an encryption scheme that allows users to perform computations on ED. Given two ciphertexts $c$ and $c'$, a user can compute $f(c,c')$ where $f$ is a function associated either with an addition or multiplication operation. A typical HE scheme consists of the following four algorithms: 


\begin{itemize}
	\item $\mathsf{HE.KeyGen}(1^{\lambda}) \rightarrow \mathsf{(pk, sk)}$: The key generation algorithm takes as input a security parameter $\lambda$ and outputs a public/private key pair ($\mathsf{pk}$, $\mathsf{sk}$).
	
	\item $\mathsf{HE.Enc(pk, }m) \rightarrow c$: This algorithm takes as input a $\mathsf{pk}$ and a message $m$ and outputs a ciphertext $c$.
	
	\item $\mathsf{HE.Eval(pk,}f,c,c^{'}) \rightarrow c_{eval}$: This algorithm takes as an input two ciphertexts $c$ and $c^{'}$, a $\mathsf{pk}$ and a homomorphic function $f$ and outputs an evaluated ciphertext $c_{eval}$. 
	
	\item $\mathsf{HE.Dec(sk, }c)\rightarrow m$: The decryption algorithm takes as input a private key $\mathsf{sk}$ and $c_{eval}$ and outputs $f(m, m')$.  
\end{itemize} 

Currently, there are three different kinds of HE schemes; Partial HE (PHE), Fully HE (FHE) and somewhat HE (SHE). 
PHE allows users to perform an \textit{unlimited number} of operations on the ciphertexts~\cite{paillier1999public, elgamal1985public}.  However, they support only one type of operation (either addition or multiplication) and hence, are not suitable for our work.  Furthermore, while FHE schemes offer the possibility to perform an unlimited number of both additions  and multiplications~\cite{gentry2009fully}, they are computationally expensive~\cite{wang2013exploring}. As a result, we choose to work with SHE that offers similar functionalities as FHE but in a more efficient manner~\cite{fan2012somewhat, acar2018survey}. The key difference between FHE and SHE is that in SHE schemes users can only perform a limited number of operations.

\section{Chebyshev Polynomials}
\label{sec:chebyshev approximation}

In this section, we show how low-degree Chebyshev polynomials can be utilized to approximate the AFs. As mentioned in~\cite{atkinson2005functional}, using these polynomials the AFs can be approximated at a given interval. The first few Chebyshev polynomials are given below while their generalization is given in equation~\ref{equ:chebyshev}.


\begin{align*} 
	T_{0}(x)=1,  
	T_{1}(x)=x, 
	T_{2}(x)=2x^{2}-1, 
	T_{3}(x)=4x^{3}-3x
\end{align*}

 

Chebyshev approximation is also known as the minimax approximation.  The minimax polynomial approach is used for function approximation by improving the accuracy and lowering the overall computational complexity~\cite{schlessman2002approximation}. Instead of minimizing the error at the point of expansion like Taylor's polynomial approximation, the minimax approach minimizes the error across a given input segment. The minimax approximation is used to find a mathematical function that minimizes the maximum error. As an example, for a function $f$ defined over the interval $[a, b]$, the minimax approximation finds a polynomial $p(x)$ that minimizes $\underset{a\le x \le b}{max}|f(x)-p(x)|$.

The first order minimax polynomial is defined as: 

\begin{align*}
	p(x)=c_0+c_1{x} \approx f(x)
\end{align*}

where $c_0$ and $c_1$ are the coefficients of the polynomial. 
%
%
%

\subsection{Chebyshev Approximation}
\label{subsec:chebyshev approx}


To approximate a continuous function $f$, defined over $[a, b]$, we first need to express $f$ as a series of Chebyshev polynomials at $[-1, 1]$. More precisely, $f$ is expressed as: $f(x) = \sum_{k=0}^{n}c_{k}T_{k}(x), \ \ x \in [-1, 1]$, where $c_{k}$ is the Chebyshev coefficient and $T_{k}(x)$ can be calculated from equation~\ref{equ:chebyshev}. As a next step, we calculate the coefficients of the polynomial and finally, express the polynomial in the original interval $[a, b]$. This procedure is illustrated in algorithm~\ref{alg:cheb poly approx}.

\begin{algorithm}
\scriptsize
	\DontPrintSemicolon
	
	\KwInput{$c_{k}, f(x), T_{k}(x)$}
	\KwOutput{$p(x)$}
	Express $f$ as: $f(x) = \sum_{k=0}^{n}c_{k}T_{k}(x), \ \ x \in [-1, 1]$\\
	Chebyshev coefficients $c_{k}$ = $\frac{2}{\pi} \int_{-1}^{1} f(x) \frac{T_{k}(x)}{\sqrt{1-x^{2}}}$\ \\
	Approximation domain: from [-1, 1] to [a, b]: $x=\frac{a+b-2z}{b-a}$, $z \in [a, b]$
	\caption{Chebyshev Polynomial Approximation}
	\label{alg:cheb poly approx}
\end{algorithm}

Our results for approximating Sigmoid and ReLU, using algorithm~\ref{alg:cheb poly approx}, are illustrated in~\autoref{tab:approximate sig and relu}. The approximation error for both AFs is calculated using equation $E(x)= f(x) - p(x)$.


\begin{table}[!ht]
\centering
		\caption{Approximating Sigmoid and ReLU Activation Function}
	\label{tab:approximate sig and relu}
 \scriptsize
	\begin{tabular}{|l |l|l|l|l|}
		\hline
        \rowcolor{Gray}
		\multicolumn{5}{|l|}{
  {\ \ \ \ \ \ \ \ \ \ \ \ \textbf{Approximation: Sigmoid}}}                \\ \hline
		x  & Interval    & Function(x) & Approximation & Difference \\ \hline
		-4 & {[}-5, 5{]} & 0.017986    & 0.016360      & -1.63e-03  \\ \hline
		-3 & {[}-5, 5{]} & 0.047426    & 0.049098      & 1.67e-03   \\ \hline
		-2 & {[}-5, 5{]} & 0.119203    & 0.118340      & -8.63e-04  \\ \hline
		-1 & {[}-5, 5{]} & 0.268941    & 0.268522      & -4.19e-04  \\ \hline
		1  & {[}-5, 5{]} & 0.731059    & 0.731478      & 4.19e-04   \\ \hline
		2  & {[}-5, 5{]} & 0.880797    & 0.881660      & 8.63e-04   \\ \hline
		3  & {[}-5, 5{]} & 0.952574    & 0.950902      & -1.67e-03  \\ \hline
		4  & {[}-5, 5{]} & 0.982014    & 0.983640      & 1.63e-03   \\ \hline
		\rowcolor{Gray} \multicolumn{5}{|l|}{
  {\ \ \ \ \ \ \ \ \ \ \ \ \textbf{Approximation: ReLU}}}                   \\ \hline
		-4 & {[}-5, 5{]} & 0.000000    & -0.008871     & -8.87e-03  \\ \hline
		-3 & {[}-5, 5{]} & 0.000000    & 0.014340      & 1.43e-02   \\ \hline
		-2 & {[}-5, 5{]} & 0.000000    & -0.015085     & -1.51e-02  \\ \hline
		-1 & {[}-5, 5{]} & 0.000000    & -0.026883     & -2.69e-02  \\ \hline
		1  & {[}-5, 5{]} & 1.000000    & 0.973117      & -2.69e-02  \\ \hline
		2  & {[}-5, 5{]} & 2.000000    & 1.984915      & -1.51e-02  \\ \hline
		3  & {[}-5, 5{]} & 3.000000    & 3.014340      & 1.43e-02   \\ \hline
		4  & {[}-5, 5{]} & 4.000000    & 3.991129      & -8.87e-03  \\ \hline
	\end{tabular}
\end{table}

%

\section{Methodology}
\label{sec:methodology}

We start this section by describing our system model. We assume a client-server model involving the following entities:

\begin{itemize}
\item \textit{Users}: We consider users who own a list of images and wish to use a cloud-based ML service to classify them in a privacy-preserving way (i.e. without revealing anything about the content of the images to the cloud). 
\item \textit{Cloud Service Provider (CSP)}: The CSP can receive a large number of \textit{encrypted} images from  users and classify them in a privacy-preserving way by giving them as input to a ML algorithm. 
\end{itemize}

The topology of our work is illustrated in~\autoref{fig:client server model}.



\begin{figure}[h]
	\centering
	\includegraphics[width=\linewidth, frame=1pt]{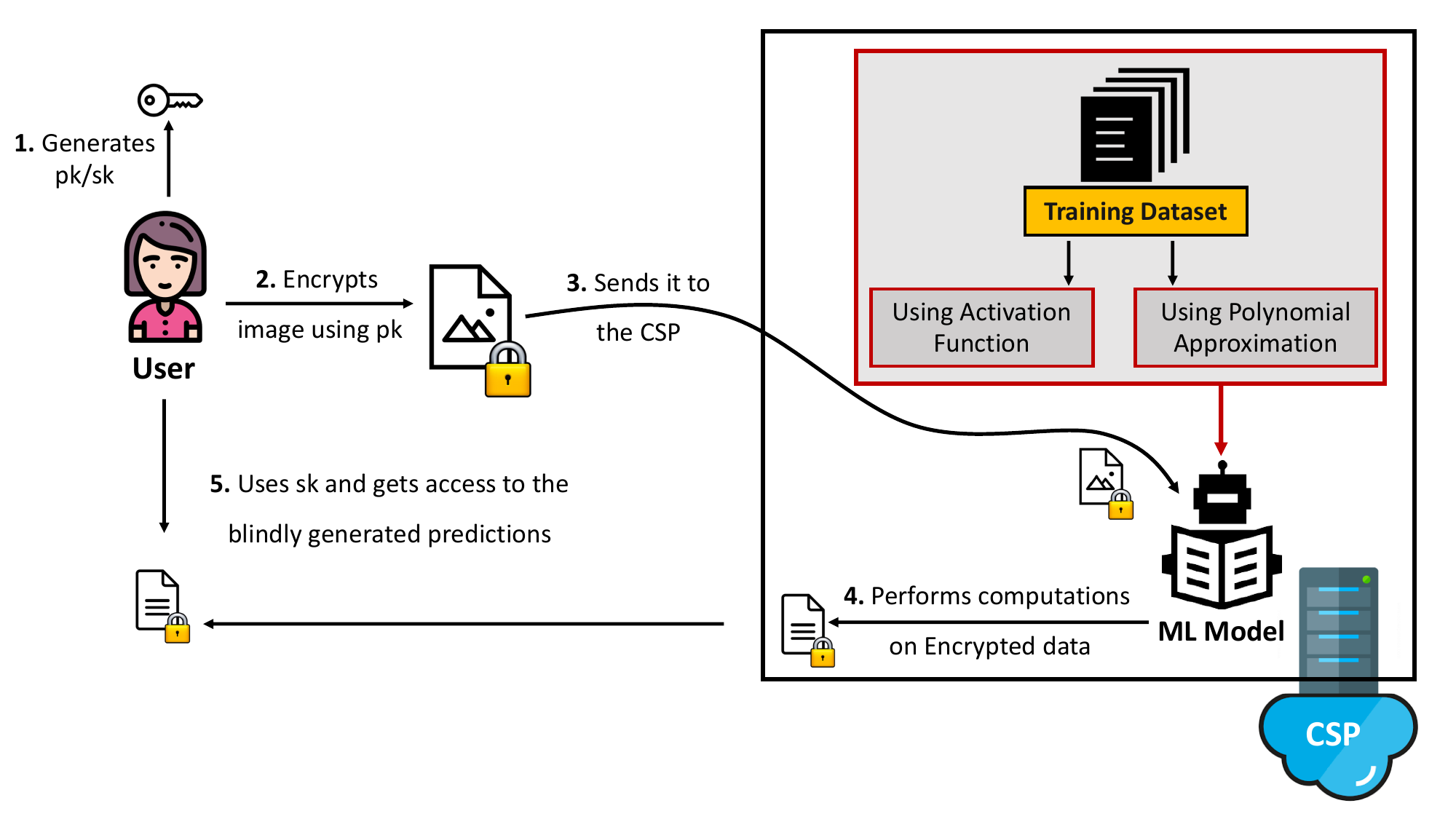}
	\caption{Learning in the Dark High Level Overview}
	\label{fig:client server model}
\end{figure}

In our model, we consider a CNN capable of analyzing large volumes of data (images) in a variety of domains. The CNN is deployed in a privacy-preserving manner in the CSP. 
To preserve the privacy of users data, we use HE. Using an HE scheme allows us to perform computations on ED. However, HE schemes face certain limitations as they only support addition and multiplication operations. Most of the operations in a CNN are simple additions and multiplications and can thus be evaluated using HE. However, AFs are non linear and as a result, we cannot use HE to perform operations on them. To this end, we replace the AFs with polynomial approximations as already discussed in section~\ref{sec:chebyshev approximation}. While higher degree polynomials would provide us better approximation, they also introduce higher computation and communication costs and hence, render our construction inefficient.

\paragraph*{\textbf{Flow}} The CNN model is deployed in the CSP and is trained using plaintext data. The weights and biases for this model are measured and made available to the CSP. For the training phase, we use the CNN given in~\autoref{fig:convolutional training}. The user generates a public/private key pair for the HE scheme, encrypts an image and sends it to the CSP. Upon reception, the CSP runs the ML model and performs the classification in a privacy-preserving way. 


\begin{figure}
\centering
\scriptsize
	\fbox{\begin{minipage}{25em}
			\begin{enumerate}
				\item \textit{\underline{Conv}}: Input image 28 $\times$ 28, Window size 5 $\times$ 5, Stride(1,1), Number of input channels 1, number of output channels 5, Filters 5, Output 28 $\times$  28 $\times$  5.
				\item \textit{\underline{Activation Function}}: ReLU .
				\item \textit{\underline{Pooling}}: Mean, Window size 2 $\times$ 2 $\times$ 1, Stride(1,1), Output 14 $\times$ 14 $\times$ 5.
				\item \textit{\underline{Conv}}: Window size 5 $\times$ 5, Stride(1,1), Filters 10, Output 14 $\times$ 14 $\times$ 10.
				\item \textit{\underline{Activation Function}}: ReLU .
				\item \textit{\underline{Pooling}}: Mean, Window size 2 $\times$ 2 $\times$ 1, Stride(1,1), Output 7 $\times$ 7 $\times$ 10.
				\item \textit{\underline{Fully Connected}}: Fully connects the incoming 7 $\times$ 7 $\times$ 10 nodes to the outgoing 128 nodes.
				\item \textit{\underline{Activation}}: ReLU 
				\item \textit{\underline{Fully Connected}}: Fully connects the incoming 128 nodes to the outgoing 10 nodes.
				\item \textit{\underline{Softmax}}: Generate a probability for 10 nodes
			\end{enumerate}
	\end{minipage}}
	\caption{Convolutional Neural Network for Training Phase }
	\label{fig:convolutional training}
\end{figure}

\subsection{Inference Phase}
\label{subsec:Inference}
Although, the operations performed in the inference phase are nearly the same as in the training phase, there are few fundamental differences. For example, all operations in the inference phase are taking place on ED while in contrast to the training phase where plaintext data is used. Similarly, the softmax which is part of the training phase is no longer available in the inference phase as shown in~\autoref{fig:convolutional inference}. 

For the inference phase, we use the Fan-Vercauteren SHE scheme~\cite{fan2012somewhat}. 
The reason for using this specific scheme is that it allows us to perform \textit{both} addition and multiplication. It is important to note that this scheme has three important parameters that affect the security level, and its performance: 

\begin{itemize}
	\item \underline{Polynomial Modulus}: This is an important parameter that affects the security level of the scheme. Polynomial modulus uses a power of two cyclotomic polynomial~\cite{thangadurai2000coefficients} and the recommended degrees for these polynomials are~1024, 2048, 4096, 8192 and beyond. On one side, a higher degree gives more security to the scheme while on the other side it degrades its performance.
	
	\item \underline{Coefficient Modulus}: This parameter determines the Noise Budget (NB) in the encrypted ciphertext. The coefficient modulus is directly proportional to NB and inversely proportional to the security level of the scheme.
	
	\item \underline{Plaintext Modulus}: The plaintext modulus affects NB in the freshly encrypted ciphertext. Additionally, it affects the NB consumption of homomorphic multiplications. For good performance, the recommendation is to keep the plaintext modulus as small as possible.
\end{itemize}

Each ciphertext in this encryption scheme has a specific quantity called NB -- measured in bits. The NB is determined by the above parameters and consumed by the homomorphic operations. The consumption of the NB is based on the chosen encryption parameters. For addition operations, this budget consumption is almost negligible in comparison to multiplication operation. In sequential multiplication that occurs at the Conv and FC layer, the consumption of NB is very high. Hence, it is important to reduce the multiplicative depth of the circuit by considering appropriate encryption parameters. Once the NB drops to zero, then the decryption of ciphertext is not possible. Therefore it is necessary to choose the parameters to be large enough to avoid this, but not so large that it becomes ineffective and non functional. 

While the HE scheme is based on polynomials, user's input is provided as a real number. Therefore, there is a clear mismatch between the two. Hence, it is important to use an encoding scheme that maps one to the other. To this end, the user encodes the input using the plaintext modulus and then encrypts it using the public key. The user generates encryption parameter and shares it with CSP. To perform computations on ED, CSP must have access to these parameters. 

Using calculated weights and biases from the training phase and encryption parameters, CSP runs the inference phase on encrypted image. The inference network is same as training network except that AFs are replaced by polynomial approximation and are built using an HE function. 

The AFs are substituted by polynomials. Since these polynomials only have addition and multiplication operations that are supported by HE. Consequently, we can perform encrypted computations on these functions. Similarly, pooling operation in the inference phase is straightforward -- calculate the average value of four ciphertexts and multiply it with appropriate values. However, Conv is a bit expensive in terms of NB as it is a sequence of multiplication operations.

Softmax is not a part of the inference network and CSP use it to perform computation on ED and obtains an encrypted output. The CSP does not have access to the secret key and thus cannot access the result. Furthermore, as the softmax layer is removed from the inference network the CSP is not able to predict the final output of the layer. 

At the end, the encrypted result of the output layer -- an array of 10 values which are homomorphically encrypted -- is sent back to the user. The user decrypts the results using the secret key and finds the output of the model which is the index corresponding to the highest among the 10 values.

\smallskip

At this point, it is important to highlight that the user utilized the ML model offered by the CSP and received the results without getting any valuable information about the underlying model. Similarly, the CSP ran the model on the encrypted image but at the same time was unable to extract any valuable information either for the content of the image or the actual prediction that sent back to the user. Hence, our model is considered as a privacy-preserving one. 

\begin{figure}
    \scriptsize
    \centering
	\fbox{\begin{minipage}{25em}
			\begin{enumerate}
				\item \textit{\underline{Conv}}: Input image 28 $\times$ 28, Window size 5 $\times$ 5, Stride(1,1), Number of input channels 1, number of output channels 5, Filters 5, Output 28 $\times$  28 $\times$  5.
				\item \textit{\underline{Activation Function}}: Approximated using polynomial approximation.
				\item \textit{\underline{Pooling}}: Mean, Window size 2 $\times$ 2 $\times$ 1, Stride(1,1), Output 14 $\times$ 14 $\times$ 5.
				\item \textit{\underline{Conv}}: Window size 5 $\times$ 5, Stride(1,1), Filters 10, Output 14 $\times$ 14 $\times$ 10.
				\item \textit{\underline{Activation Function}}: Approximated using polynomial approximation.
				\item \textit{\underline{Pooling}}: Mean, Window size 2 $\times$ 2 $\times$ 1, Stride(1,1), Output 7 $\times$ 7 $\times$ 10.
				\item \textit{\underline{Fully Connected}}: Fully connects the incoming 7 $\times$ 7 $\times$ 10 nodes to the outgoing 128 nodes.
				\item \textit{\underline{Activation}}: Approximated using polynomial approximation. 
				\item \textit{\underline{Fully Connected}}: Fully connects the incoming 128 nodes to the outgoing 10nodes.
			\end{enumerate}
	\end{minipage}}
	\caption{Convolutional Neural Network for Inference Phase }
	\label{fig:convolutional inference}
\end{figure}

\section{Performance Evaluation}
\label{sec:performance analysis}
We present our experimental results. In the first part, we provide experimental results on function approximation using Chebyshev polynomials. Then, we evaluate the performance of the proposed ML model and compare it with CryptoNETs. 


\paragraph{\textbf{Experimental Setup}} 
All experiments were conducted in Python~3 using Ubuntu~18.04 LTS~64 bit (Intel Core~i7,~2.80 GHz,~32GB). For the training phase, we used \href{https://www.tensorflow.org/}{Tensor flow} to train our CNN model, while the actual experiments for that phase were conducted on \href{https://colab.research.google.com/}{Google Colab} (with GPU enabled).
Finally, for the inference phase we used Microsoft's Simple Encrypted Arithmetic Library (SEAL)~\cite{sealcrypto}. 


\paragraph{\textbf{Dataset}} To evaluate our model, similar to other works in the area, we used the MNIST dataset~\cite{lecun-mnisthandwrittendigit-2010} which consists of~60,000 images of handwritten digits. To train our CNN model we used 50,000 images while the rest 10,000 were used for testing. Each image is $28 \times 28$ pixel array and is represented by its gray level in the range of~0-255.


\subsection{Activation Function Approximation}
\label{subsec:ActiFuncApprox}

As we mentioned in the previous sections, in our approach we use Chebyshev polynomials to approximate the ReLU and Sigmoid AFs where inputs are images encrypted with an SHE scheme. The polynomial approximation of the ReLU AF is shown in~\autoref{tab:relu-polyapprox_degree} while for Sigmoid in~\autoref{tab:sig-polyapprox_degree}. Since the choices of the degree and the interval affect the performance of the model, it is necessary to choose suitable parameters. For this purpose, we conducted a series of experiments using different degrees and intervals. As can be seen in~\autoref{tab:relu-polyapprox_degree}, the AFs are more accurately approximated when using polynomials of higher degree in small intervals. For example, the polynomial having degree~9 and interval $[-10, 10]$ more accurately approximate the ReLU function than the rest of the polynomials. The same applies to the Sigmoid AF, where a high degree~9  and small interval $[-10, 10]$ give a better approximation as can be seen in~\autoref{tab:sig-polyapprox_degree}. However, the use of higher degree polynomials introduces a significant computation overhead, and small intervals limit the use of the approximation function. 
The results for approximating the Sigmoid AF using low-degree Chebyshev polynomials are presented in~\autoref{tab:sig-polyapprox_degree}.

		\begin{table*}[!ht]
  \centering
			\caption{Polynomial Approximation of the ReLU Function on Two Intervals ([-10,10], [-100,100]) using Different Degrees} 
			\label{tab:relu-polyapprox_degree}
            \scriptsize
			\scalebox{0.8}{
			\begin{tabular} {|p{9mm}|p{15mm}|>{\small}p{90mm}|p{25mm}|}  \hline
                \rowcolor{Gray}
				Degree & Interval &Polynomial Approximation & ReLU Function \\ \hline
				
				3 &\(\displaystyle [-10, 10] \) &\(\displaystyle (-4.44089209850063e-18) \times x^3 + (0.038268343236509) \times x^2 + (0.5) \times x + 1.35299025036549\) & \parbox[c]{1em}{
					\includegraphics[width=1in]{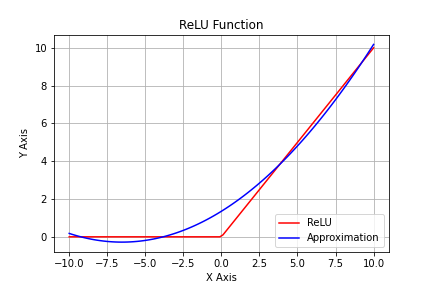}} \\ \hline 
				
				5 &\(\displaystyle [-10, 10] \) & \(\displaystyle (2.368475785867e-19) \times x^5 - (0.000252624921308674) \times x^4 - (2.90138283768708e-17) \times x^3 + (0.0660873211772537) \times x^2 + (0.5) \times x + 0.862730150341736 \) & \parbox[c]{1em}{
					\includegraphics[width=1in]{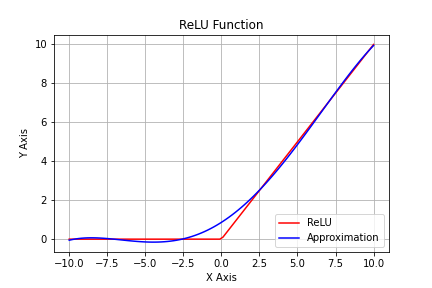}} \\ \hline 
				
				7 &\(\displaystyle [-10, 10] \) & \(\displaystyle (-8.88178419700125e-21) \times x^7 + (3.66197231323541e-6) \times x^6 + (1.33226762955019e-18) \times x^5 - (0.000847927183186682) \times x^4 - (5.24025267623074e-17) \times x^3 + (0.0920352084972136) \times x^2 + (0.5) \times x + 0.637244473880199 \) & \parbox[c]{1em}{
					\includegraphics[width=1in]{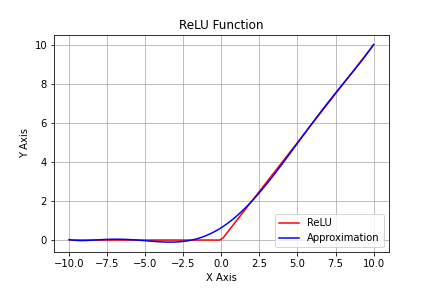}} \\ \hline 
				9 &\(\displaystyle [-10, 10] \) & \(\displaystyle (1.15960574476048e-21) \times x^9 - (7.03111115816643e-8) \times x^8 - (2.41868747252738e-19) \times x^7 + (1.87324195121527e-5) \times x^6 + (1.66338054441439e-17) \times x^5 - (0.00189480875502865) \times x^4 - (4.21263024463769e-16) \times x^3 + (0.117284304779533) \times x^2 + (0.5) \times x + 0.506232562894004 \) & \parbox[c]{1em}{
					\includegraphics[width=1in]{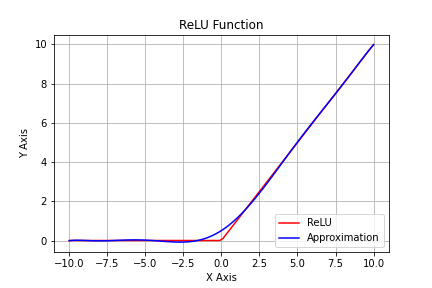}} \\ \hline 
				
				3 &\(\displaystyle [-100, 100] \) & \(\displaystyle (-4.2632564145606e-20) \times x^3 + (0.0038268343236509) \times x^2 + (0.5)\times x + 13.5299025036549 \) & \parbox[c]{1em}{
					\includegraphics[width=1in]{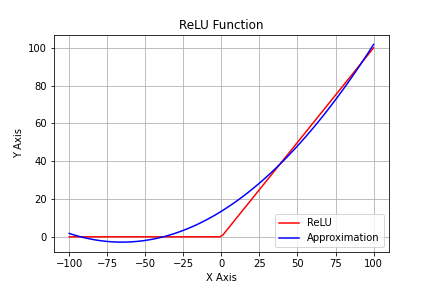}} \\ \hline 
				
				5 &\(\displaystyle [-100, 100] \) & \(\displaystyle (2.27373675443232e-23) \times x^5 - (2.52624921308674e-7) \times x^4 - (2.70006239588838e-19) \times x^3 + (0.00660873211772537) \times x^2 + (0.5)\times x + 8.62730150341737 \) & \parbox[c]{1em}{
					\includegraphics[width=1in]{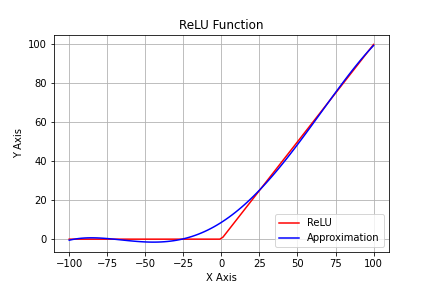}} \\
				\hline
				
				7 &\(\displaystyle [-100, 100] \) & \(\displaystyle (-6.82121026329696e-27) \times x^7 + (3.6619723132354e-11) \times x^6 + (1.03739239420975e-22) \times x^5 - (8.47927183186682e-7) \times x^4 - (4.2277292777726e-19) \times x^3 + (0.00920352084972135) \times x^2 + (0.5) \times x + 6.37244473880199 \) & \parbox[c]{1em}{
					\includegraphics[width=1in]{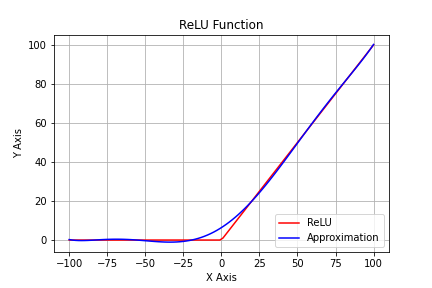}} \\
				\hline
				
				9 &\(\displaystyle [-100, 100] \) & \(\displaystyle (1.12777343019843e-29) \times x^9 - (7.03111115816644e-15) \times x^8 - (2.35559127759188e-25) \times x^7 + (1.87324195121527e-10) \times x^6 + (1.62231117428746e-21) \times x^5 - (1.89480875502865e-6) \times x^4 - (4.11404244005098e-18) \times x^3 + (0.0117284304779533) \times x^2 + (0.5) \times x + 5.06232562894004 \) & \parbox[c]{1em}{
					\includegraphics[width=1in]{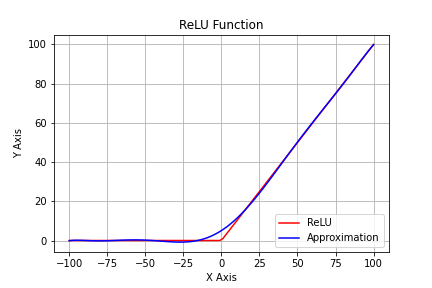}} \\
				\hline 
			\end{tabular}
		}
		\end{table*}

\begin{table*}[h!]
\centering
	\caption{Polynomial Approximation of the Sigmoid Function on Two Intervals ([-10,10], [-100,100]) using Different Degrees} 
	\label{tab:sig-polyapprox_degree}
    \scriptsize
	\scalebox{0.8}{
	\begin{tabular} {|p{9mm}|p{15mm}|>{\small}p{90mm}|p{25mm}|}  \hline
		\rowcolor{Gray}	
		Degree & Interval &Polynomial Approximation & Sigmoid Function \\ \hline
		
		3 &\(\displaystyle [-10, 10] \) &\(\displaystyle (-0.00100377373568484) \times x^3 + (1.45518367592346e-13) \times x^2 + (0.139786538317376) \times x + 0.499999999992724 \) & \parbox[c]{1em}{
			\includegraphics[width=1in]{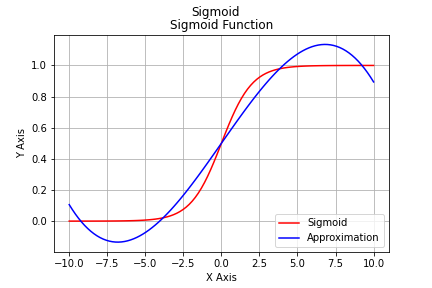}} \\ \hline 
		5 &\(\displaystyle [-10, 10] \) & \(\displaystyle (2.0467424332792e-5) \times x^5 + (5.82078097744257e-15) \times x^4 - (0.00336794817488311) \times x^3 - (5.65619279205865e-13) \times x^2 + (0.187819515164365) \times x + 0.500000000006453 \) & \parbox[c]{1em}{
			\includegraphics[width=1in]{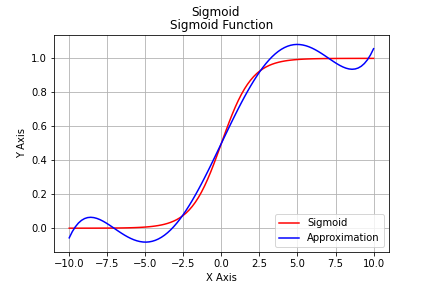}} \\ \hline 
		
		7 &\(\displaystyle [-10, 10] \) & \(\displaystyle (-4.34913635838155e-7) \times x^7 - (5.82079696725621e-16) \times x^6 + (9.18419138902492e-5) \times x^5 + (8.44014964905896e-14) \times x^4 - (0.00652613009889838) \times x^3 - (3.00134166245124e-12) \times x^2 + (0.216030242339756) \times x + 0.500000000015461 \) & \parbox[c]{1em}{
			\includegraphics[width=1in]{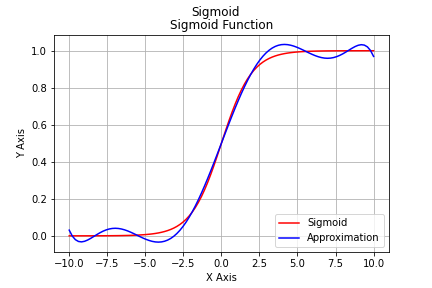}} \\ \hline 
		9 &\(\displaystyle [-10, 10] \) & \(\displaystyle (9.32721914680041e-9) \times x^9 + (1.39698499452418e-17) \times x^8 - (2.42773327147286e-6) \times x^7 - (2.60854055994519e-15) \times x^6 + (0.000229352354062705) \times x^5 + (1.47390762630842e-13) \times x^4 - (0.0097848700927233) \times x^3 - (2.49775180627496e-12) \times x^2 + (0.231624826001611) \times x + 0.500000000005353 \) & \parbox[c]{1em}{
			\includegraphics[width=1in]{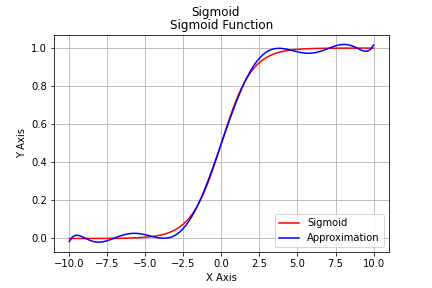}} \\ \hline 
		
		3 &\(\displaystyle [-100, 100] \) & \(\displaystyle (-1.082392200292393945e-6) \times x^3 + 6.9623566103164815542e-21\times x^2 + 0.014650756326574837423\times x + 0.49999999999999996519 \) & \parbox[c]{1em}{
			\includegraphics[width=1in]{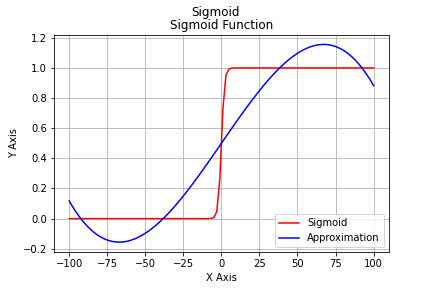}} \\ \hline 
		
		5 &\(\displaystyle [-100, 100] \) & \(\displaystyle (2.76073648126615e-10)\times x^5 + (5.82075253629132e-19)\times x^4 - (4.39372964314194e-6)\times x^3 - (5.65616515992073e-15)\times x^2 + 0.0221378748242352\times x + 0.500000000006453 \) & \parbox[c]{1em}{
			\includegraphics[width=1in]{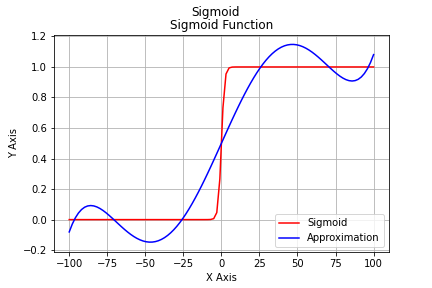}} \\
		\hline
		
		7 &\(\displaystyle [-100, 100] \) & \(\displaystyle (-8.15672916212668e-14)\times x^7 - (5.82076636528339e-22)\times x^6 + (1.66796555589344e-9)\times x^5 + (8.44011305257745e-18)\times x^4 - (1.10438386427853e-5)\times x^3 - (3.00133479250796e-14)\times x^2 + 0.029595343798993\times x + 0.500000000015461 \) & \parbox[c]{1em}{
			\includegraphics[width=1in]{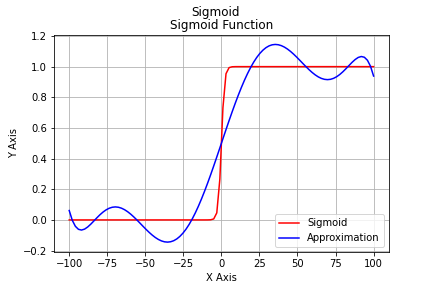}} \\
		\hline
		
		9 &\(\displaystyle [-100, 100] \) & \(\displaystyle (2.59190909648308e-17)\times x^9 + (1.3969822658824e-25)\times x^8 - (6.55008389245384e-13)\times x^7 - (2.60853504193157e-21)\times x^6 + (5.85712544194627e-9)\times x^5 + (1.47390439498889e-17)) x^4 - (2.21440994162872e-5)\times x^3 - (2.49774773767457e-14)\times x^2 + 0.0370400456474675\times x + 0.500000000005353 \) & \parbox[c]{1em}{
			\includegraphics[width=1in]{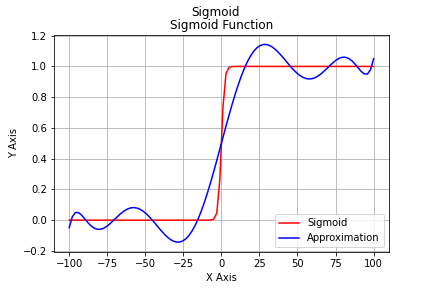}} \\
		\hline 
	\end{tabular}
}
\end{table*}

Furthermore, we performed a plethora of different experiments on the CNN model. We trained different networks by increasing the size of the Conv and the size of the filters. We noticed that changing the number of Conv and filters affects the overall accuracy of the network. As the size of the filter and layer increases, the accuracy of the network also increases. However, the efficiency of the network drops significantly. Hence, for the training phase, we considered the network given in~\autoref{fig:convolutional training}. First, we trained the CNN model using the ReLU AF. The measured accuracy for that part was~99.2\%. Then the same network was trained using the polynomial approximation function where we got an accuracy of~98.5\% -- a result that is very close to the original AF.

For comparison, we used the model proposed in CryptoNets~\cite{gilad2016cryptonets} which is similar to the one proposed in our paper -- a Conv, FC layers and an average polling layer as shown in~\autoref{tab:comp prev model}. Training the model with the ReLU AF, the accuracy of our model was~99.20\% whereas CryptoNets achieved a~99\%. Similarly, for the approximated function we obtained an accuracy of~98.5\% while CryptoNets achieved~98.95\%. For the same network, the accuracy of the model proposed in~\cite{hesamifard2018privacy} was~99.02\% using the ReLU AF and~99\% using the approximated function. 

\subsection{Performing Computation on Encrypted Data}
\label{subsec:HomoEvalCNN}
Now, we proceed by discussing how the use of HE can affect the performance of the NN model. In our work, we trained the CNN model on plaintext data while the classification was performed on the ciphertexts. As a result, we had to perform computations on two types of data -- plaintext and ciphertext. For this purpose, we used the \href{https://github.com/Lab41/PySEAL}{SEAL library} that allowed us to perform computations on ciphertext. Although the use of SEAL is straightforward, we still had to define certain parameters (see Section~\ref{subsec:Inference}). 

We performed a series of experiments using different encryption parameters. First, we looked at the polynomial modulus -- the encryption parameter used in SEAL. 
During the experiments we observed that a smaller value of polynomial modulus leads to a more efficient result but at the same time the accuracy is decreased. In contrast, a higher value of the polynomial modulus gives more accurate results, however, degrades the performance. The second encryption parameter is the coefficient modulus that decides the NB in the freshly encrypted ciphertext. This parameter is automatically set by SEAL based on the value of security level and polynomial modulus. Finally, increasing the value of the plaintext modulus, decreases the consumption of the NB.

\subsection{Comparison with the Existence Model}
\label{subsec:comparison with existence model}
Finally, we compared our results with state-of-the-art privacy-preserving NNs that utilize HE. The work proposed in CryptoNets~\cite{gilad2016cryptonets} is similar to ours. In CryptoNets, the model is trained on plaintext data and then the trained model is used for the classification of encrypted instances. In order to have a fair comparison, it is important to incorporate the same network used in both works. To this end, we used the CryptoNets model. 
Instead of using the overall performance of the model we decided to equate each layer. As can be seen in~\autoref{tab:comp prev model}, our model outperforms CryptoNets at both the encryption and decryption times as well as in the activation layer.

\begin{table}
	\centering
	\caption{Comparison with the Previous Models} 
	\label{tab:comp prev model}
 \scriptsize
	\begin{tabular}{|p{18mm}|p{25mm}|p{12mm}|p{15mm}|}
		\hline	
		\rowcolor{Gray}		
		\textbf{Layer} & \textbf{Description} & \multicolumn{2}{c|}{\textbf{Time}} \\  \hline
		\rowcolor{Gray}
		& & CryptoNets     & 
  Learning in the Dark \\ \hline
		Encryption                                                           & Encoding+Encryption                                                                & 44.5        & 8.5242      \\ \hline
		1st Conv                                                & Same except stride value                                                    & 30          & 60.36       \\ \hline
		1st AF                                              & \begin{tabular}[c]{@{}l@{}}NLAF\end{tabular} & 81          & 6.62        \\ \hline
		1st pooling layer                                                    & Mean pooling                                                             & 127         & 0.188       \\ \hline
		2nd Conv                                              & -                                                                                  & -           & 64.822      \\ \hline
		2nd AF                                              & \begin{tabular}[c]{@{}l@{}}NLAF\end{tabular} & 10          & 0.199       \\ \hline
		2nd pooling layer                                                    & -                                                                                  & -           & 0.092       \\ \hline
		\begin{tabular}[c]{@{}l@{}}1st FC layer\end{tabular}  & Generates 10 output                                                                & 1.6         & 12.1839     \\ \hline
		\begin{tabular}[c]{@{}l@{}}2nd FC layer\end{tabular} & -                                                                                  & -           & 0.326       \\ \hline
		Decryption                                                           & Image decryption                                                                   & 3           & 0.0021      \\ \hline
	\end{tabular}
\end{table}


\section{Learning in the Dark Protocol}
\label{sec:PPMLProtocol}
In the first part of this section, we formalize the communication between the user and the CSP by designing a detailed protocol. Then, we prove the security of our construction in the presence of a malicious adversary. For the rest of the section, we assume the existence of a cryptogrpahic hash function that is first and second pre-image resistant. Before we proceed to the formal description of our protocol, we present a high-level overview of our construction.


\paragraph*{\textbf{High-Level Overview}} We assume that a user $u$ wishes to classify an image in a privacy-preserving way. To this end, $u$ first outputs an image and encrypts it using an HE scheme. As a next step, $u$ sends the encrypted image to the CSP. Upon reception, the CSP commences the classification process directly on the encrypted image without the need to decrypt it. To achieve this, the CSP runs the evaluation algorithm of the HE scheme on the encrypted image and finally, outputs an encrypted vector. Each element of the vector represents the probability that the image belongs to a certain class. Finally, the CSP sends the encrypted vector back to $u$. Upon reception, $u$ decrypts the vector and classifies her image to the class that has the highest probability.  

\subsection{Construction}
\label{subsec:protocol}

As already stated in Section~\ref{sec:methodology}, we assume a client-server model. Our protocol takes part in two different phases; a \textit{Setup} phase and a \textit{Running} phase.

\paragraph*{\textbf{Setup Phase}} In the Setup phase, the user $u$ and the CSP establish a shared symmetric key $\mathsf{K}$. This key will be used to secure the communication between the two entities. Apart from that $u$ also generates a public/private key pair for a HE scheme. More specifically, $u$ executes $\mathsf{HE.KeyGen}(1^\lambda) \rightarrow (\mathsf{pk, sk})$, for some $\lambda$. We assume that upon its generation, $\mathsf{pk}$ is publicly known while $\mathsf{sk}$ remains secret.  

\paragraph*{\textbf{Running Phase}} After the successful execution of the \textit{Setup} phase, $u$ can start communicating with the CSP. To do so, $u$ first encrypts an image $img$ by running $\mathsf{HE.Enc(pk,}img) \rightarrow c_{img}$. Moreover, $u$ generates an unpredictable random number $r_1$ and sends to the CSP $m_1 = \langle r_1, c_{img}, \mathsf{HMAC}(\mathsf{K}, r_1||c_{img}) \rangle$. Upon reception, the CSP checks the freshness of the message by looking at $r_1$, and verifies the $\mathsf{HMAC}$ using the shared key $\mathsf{K}$. If any of the above verifications fail, the CSP will output $\perp$ and abort the protocol. Otherwise, the CSP proceeds with the execution of the ML model described in Section~\ref{sec:methodology}. In particular, the CSP starts running $\mathsf{HE.Eval}$ and finally, outputs an encrypted vector $c_{eval}$. The encrypted vector is then sent back to $u$ via $m_2 = \langle r_2, c_{eval}, \mathsf{HMAC(K},r_2||c_{eval}|c_{img}|)$. Upon receiving $m_2$, $u$ verifies both the freshness of the message and the $\mathsf{HMAC}$. If the verification fails, $u$ outputs $\perp$ and aborts the protocol. Otherwise, $u$ decrypts $c_v$ by running $\mathsf{HE.Dec}(\mathsf{sk}, c_{eval}) \rightarrow v$. Having the plaintext vector at her disposal, $u$ can now classify her image in accordance with the probabilities included in the vector. Our construction is illustrated in~\autoref{fig:RP}

\begin{figure}
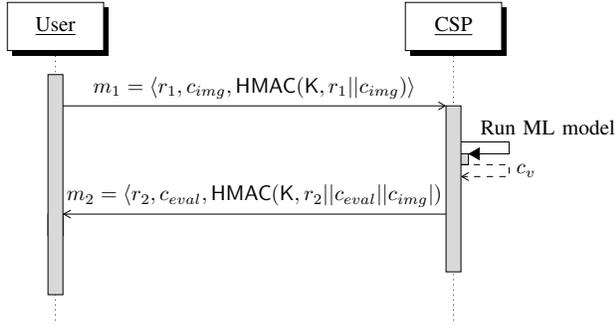

  \centering
  \scalebox{0.8}{
  \begin{sequencediagram}
    \newthread{A}{User}{}
    \newinst[5]{B}{CSP}{}
    
    \begin{messcall}{A}{$m_1 = \langle r_1, c_{img}, \mathsf{HMAC}(\mathsf{K}, r_1||c_{img}) \rangle$}{B}{}
      \begin{callself}{B}{Run ML model}{$c_v$}
	\end{callself}
	\begin{messcall}{B}{$m_2 = \langle r_2, c_{eval}, \mathsf{HMAC(K},r_2||c_{eval}||c_{img}|)$}{A}{}
	
    \end{messcall}
    \end{messcall}
  \end{sequencediagram}}
  \caption{Running Phase.}
  \label{fig:RP}
\end{figure}

\subsection{Security Analysis}
\label{subsec:security analysis}

We prove security of our protocol in presence of malicious adversary $\mathcal{ADV}$. We start by defining the threat model:

\paragraph*{\textbf{Threat Model}} Our threat model is similar to the one described in~\cite{Michalas:17:Trusted:Launch}, which is based on the Dolev-Yao adversarial model~\cite{dolev:1983}. Moreover, we extend the above threat model by defining a set of attacks available to $\mathcal{ADV}$.

\begin{myAttack}[Image Substitution Attack (ISA)] Let $\mathcal{ADV}$ be an adversary that overhears the communication between user and CSP. $\mathcal{ADV}$ successfully launches an ISA if she manages to replace encrypted image sent from  user to CSP, in a way that is indistinguishable from CSP.  
\end{myAttack}

\begin{myAttack}[Vector Substitution Attack (VSA)] Let $\mathcal{ADV}$ be an adversary that overhears the communication between the user and the CSP. $\mathcal{ADV}$ successfully launches a VSA, if she manages to replace the encrypted vector sent from the CSP to the user, in a way that is indistinguishable for the user. 
\end{myAttack}

We now proceed with proving that our protocol is secure against the defined threat model.

\begin{proposition}[Image Substitution Attack Soundness]
Let $\mathcal{ADV}$ be a malicious adversary. Then $\mathcal{ADV}$ cannot successfully launch an ISA.
\end{proposition}

\begin{proof}
For $\mathcal{ADV}$ to successfully launch an ISA, she needs to tamper with $m_1 = \langle r_1, c_{img}, \mathsf{HMAC}(\mathsf{K},r_1||c_{img})$. To do so, $\mathcal{ADV}$ has two options:

\begin{enumerate}
	\item Forge a new $m_1$ message.
	\item Replay an old $m_1$ message.
\end{enumerate}

We will show that in both cases, $\mathcal{ADV}$ can successfully launch her attack with negligible probability.

\begin{itemize}
\item Since we assume that the $\mathsf{pk}$ of the HE scheme is publicly known, $\mathcal{ADV}$ can generate a valid ciphertext $c_{img}'$ that is indistinguishable from the original $c_{img}$. As a next step, $\mathcal{ADV}$ replaces the original $c_{img}$ with the newly generated $c_{img}'$ and forwards $m_1' = \langle r_1, c_{img}', \mathsf{HMAC}(\mathsf{K}, r_1||c_{img})$ to the CSP. Upon reception, the CSP will try to verify the $\mathsf{HMAC}$. However, as $c_{img}' \neq c_{img}$ the verification will fail, and the CSP will abort the protocol. Hence, $\mathcal{ADV}$ also needs to forge a valid $\mathsf{HMAC}$. However, as $\mathcal{ADV}$ does not possess the shared key $\mathsf{K}$, this can only happen with negligible probability and thus, the attack fails.   

\item The only other alternative for $\mathcal{ADV}$, is to replay an older message. To do so, $\mathcal{ADV}$ replaces the $m_1$ message sent from the user to the CSP with an older $m_1'$ message from a previous session. Upon receiving $m_1'$, the CSP will verify the validity of the $\mathsf{HMAC}$ but it will fail to verify the freshness of the message. An alternative for $\mathcal{ADV}$, would be to generate a fresh random number and to replace the old one. However, since the random number is also included in the $\mathsf{HMAC}$, $\mathcal{ADV}$ would also need to forge a valid $\mathsf{HMAC}$. Given the fact that $\mathcal{ADV}$ does not possess the shared key $\mathsf{K}$, this can only happen with negligible probability and thus, the attack fails.
\end{itemize}

\end{proof}

\begin{proposition}[Vector Substitution Attack Soundness]
Let $\mathcal{ADV}$ be a malicious adversary. Then $\mathcal{ADV}$ cannot successfully launch a VSA.
\end{proposition}

\begin{proof}
The proof is omitted as it is similar to the previous one. More specifically, the security properties of the $\mathsf{HMAC}$, and the fact that $\mathcal{ADV}$ does not know the shared $\mathsf{K}$ are enough to ensure that $\mathcal{ADV}$ cannot successfully launch a VSA. 
\end{proof}
\noindent \textbf{Open Science and Reproducible Research:} 
To support open science and reproducible research, and provide researchers with the opportunity to use, test, and hopefully extend our work, our source code has been made available online\footnote{\href{https://gitlab.com/nisec/blind\_faith}{https://gitlab.com/nisec/blind\_faith}}.

\section{Conclusion}
\label{sec:Impact}



Undoubtedly, ML models and their underlying applications are driving the big-data economy. However, in practice, the systems using these models incorporate proxies. Many existing systems can introduce biases or rely on proxies like gender or race, leading to unfair outcomes. With this work, we aim to create a more equitable and unbiased approach to decision-making.
Learning in the Dark allows us to apply ML models directly to encrypted data so the information remains secure. We accomplished this by estimating the behavior of activation functions, which are components of ML models. Our experiments and evaluations showed promising results, demonstrating that Learning in the Dark can effectively analyze encrypted data while maintaining high accuracy.
We believe this research can inspire further advancements in privacy-preserving machine learning and contribute to systems that promote fairness, privacy, and transparency in our increasingly data-driven world.

\bibliographystyle{IEEEtran}
\bibliography{sample-base}

\end{document}